\documentclass[11pt]{article}
\usepackage{amsfonts,amsthm,amsmath}

\newtheorem{Theorem}{Theorem}[section]
\newtheorem{Definition}{Definition}[section]
\newtheorem{Example}{Example}[section]
\newtheorem{Construction}{Construction}[section]
\newtheorem{Remark}{Remark}[section]
\newtheorem{Lemma}[Theorem]{Lemma}
\newtheorem{Corollary}[Theorem]{Corollary}
\newtheorem{Conjecture}{Conjecture}

\newcommand{\eff}{\ensuremath{\mathbb F}}

\title{Optimal Ramp Schemes and Related Combinatorial Objects}
\author{Douglas~R.~Stinson\thanks{Research supported by  NSERC discovery grant RGPIN-03882}\\
David R.\ Cheriton School of Computer Science\\ University of Waterloo\\
Waterloo, Ontario N2L 3G1, Canada}
\date{\today}

\begin{document}
\maketitle

\begin{abstract}
In 1996, Jackson and Martin \cite{JM1} proved that a strong ideal ramp scheme is
equivalent to an orthogonal array. However, there was no good characterization
of ideal ramp schemes that are not strong. Here we show the equivalence of
ideal ramp schemes to a new variant of orthogonal arrays that we term
\emph{augmented orthogonal arrays}. We give some constructions for these new
kinds of arrays, and, as a consequence, we also provide parameter situations where 
ideal ramp schemes exist but strong ideal ramp schemes do not exist. 
\end{abstract}

\section{Introduction}
\label{intro.sec}

Informally, a \emph{$(t,n)$ threshold scheme} (\cite{Blakley,Shamir}) is a method
of distributing secret information (called \emph{shares}) to $n$ players,
in such a way that any $t$ of the $n$ players can compute a predetermined \emph{secret}, but no 
subset of $t-1$ players can determine the secret.  The integer $t$ is called the 
\emph{threshold}; we assume that $1 \leq t \leq n$.

It is well-known that 
the number of possible shares in a threshold scheme must be greater than 
or equal to  the number of possible secrets.
If the number of possible secrets in a threshold scheme equals the
number of possible shares, the scheme is termed \emph{ideal}.

\emph{An  $(s,t,n)$ ramp scheme} (\cite{BM}) is a generalization of a threshold scheme in which 
there are two thresholds. The value $s$ is the \emph{lower threshold} and
$t$ is the \emph{upper threshold}. In a ramp scheme, any $t$ of the $n$ players can compute
the secret (exactly as in a $(t,n)$ threshold scheme). It is also required that
no subset of $s$ players can determine the secret. We note that a $(t-1,t,n)$ ramp scheme
is exactly the same thing as a $(t,n)$ threshold scheme. The parameters of a ramp scheme
satisfy the conditions $0 \leq s < t \leq n$.

A ramp scheme with $s < t-1$ possibly permits a larger number of possible secrets for a given number of
shares than is the case in a threshold scheme. If there are $v$ possible shares in 
an $(s,t,n)$ ramp scheme, then the number of possible secrets is bounded above by 
$v^{t-s}$. Of course, for $s < t-1$, it holds that $v^{t-s} > v$. When $s = t-1$,
the bound is equal to $v$, agreeing with the above-mentioned bound for threshold schemes.
If an  $(s,t,n)$ ramp scheme can be constructed with $v^{t-s}$ possible secrets (given $v$ possible shares), then
we say that the ramp scheme is \emph{ideal}. Thus, an ideal 
$(t-1,t,n)$ ramp scheme is the same thing as an ideal $(t,n)$ threshold scheme.

One of the very first constructions for threshold schemes, the Shamir threshold scheme \cite{Shamir},
yields ideal schemes. It is also well-known that ideal threshold schemes are equivalent to 
certain well-studied combinatorial structures, namely, orthogonal arrays and maximum distance
separable (MDS) codes \cite{Mar,DMR,BK}. 

There is less work on combinatorial characterizations of optimal 
ramp schemes. The main result in this direction is due to Jackson and Martin \cite[Theorem 9]{JM1}, 
who show that an ideal $(s,t,n)$ ramp scheme that satisfies certain additional conditions 
(they call such a scheme a \emph{strong} ramp scheme) is equivalent to an ideal
$(t,n+t-s-1)$ threshold scheme. This result is perhaps not completely satisfying because the
additional conditions used to define strong ramp schemes are rather restrictive. In \cite{JM1},
the authors ask if it is possible to construct ideal ramp schemes that are not strong. 
This is one of the open questions that we answer in this paper.

Our  approach is to define a new type of combinatorial structure that we term 
an \emph{augmented orthogonal array}, or \emph{AOA}. We prove that any optimal 
ramp scheme is  equivalent to a certain augmented orthogonal array.
This equivalence can be proven in a straightforward manner, analogous to the proof that
an ideal threshold scheme is equivalent to an orthogonal array.
We then investigate some methods of constructing augmented orthogonal arrays. There is a natural
way to construct augmented orthogonal arrays from orthogonal arrays. (Roughly
speaking, the resulting augmented orthogonal arrays correspond
to the strong ramp schemes considered in \cite{JM1,Yamamoto}.) 
However, we observe that there are also constructions of augmented orthogonal arrays which yield 
ideal ramp schemes that are not strong. Moreover, we show there are parameter situations
for which there exist ideal ramp schemes, but where there do not exist strong ideal ramp schemes.
These results provide  answers to the questions that were first posed in \cite{JM1}.

For future reference, Figure \ref{fig1} shows the relationships between the ramp schemes and combinatorial
structures we discuss in this paper. 

\begin{figure}
\begin{center}
\begin{tabular}{ccc}
strong $(s,t,n)$ ramp scheme defined over a set of $v$ shares & $\Longleftrightarrow$ & OA$(t,n+t-s,v)$ \\
$\Downarrow$ & & $\Downarrow$ \\
$(s,t,n)$ ramp scheme defined over a set of $v$ shares & $\Longleftrightarrow$ & AOA$(s,t,n,v)$ 
\end{tabular}
\end{center}
\caption{Relationships between  ramp schemes and combinatorial
structures}
\label{fig1}
\end{figure}

The rest of this paper is organized as follows. In Section \ref{defns.sec}, we give formal 
definitions of ramp and threshold schemes, based on the  ``distribution rules''
for the scheme. Section \ref{equiv.sec} reviews combinatorial structures equivalent
to ideal threshold schemes (orthogonal arrays, MDS codes, etc.) in both the linear and
general cases (in this context, the term ``linear'' means that the object in question
can be viewed as a subspace of a vector space over a finite field). Section \ref{AOA.sec}
introduces the new notion of an augmented orthogonal array (AOA). We then discuss the connection
between AOAs and orthogonal arrays. We also provide some constructions for AOAs in situation
where ``associated'' orthogonal arrays do not exist. Section \ref{ramp-equiv.sec} gives the proof that
an ideal ramp scheme is equivalent to an AOA. We also provide examples of ideal ramp schemes
that are not strong in this section. Finally, we conclude with some discussion and comments in 
Section \ref{summary.sec}.

\section{Formal Definitions of Ramp and Threshold Schemes}
\label{defns.sec}

In this section, we provide formal definitions of ramp schemes in two 
flavours, namely, ``weak'' and ``perfect''. Our definitions are phrased
in terms of ``distribution rules'', which is one of the standard ways
of defining these types of schemes. (For a discussion of this model in 
relation to other models, we refer the reader to \cite{JM2}.)

Suppose there is an $(s,t,n)$ ramp scheme defined over
a set of $v$ secrets. We will assume without loss of generality that the set of possible shares
for any player is $\mathcal{X} = \{1, \dots, v\}$, and we denote the set of possible secrets 
by $\mathcal{K}$.

We now present a  formal mathematical model for a ramp scheme.
Denote the set of $n$ \emph{players} by
$\mathcal{P} = \{P_1, \dots , P_n\}$.  
A \emph{distribution rule} $d$ represents a possible distribution of shares
to the $n$ players. So we can view $d$ as a function, i.e.,
$d : \mathcal{P} \rightarrow \mathcal{X}$.
The share given to $P_i$ is $d(P_i)$, $1 \leq i \leq n$.
We may also represent $d$ as an $n$-tuple $(d_1, \dots , d_n)$, where
$d_i = d(P_i)$, for $1 \leq i \leq n$.
Finally, for a distribution rule $d$ and a subset of players $\mathcal{P}_0 \subseteq \mathcal{P}$, 
we define the \emph{projection of $d$ to $\mathcal{P}_0$}, denoted  $d|_{\mathcal{P}_0}$, 
to be the restriction of
$d$ to the subdomain $\mathcal{P}_0$. A projection can  be represented 
as a tuple of length $|\mathcal{P}_0|$.

For every possible secret $K \in \mathcal{K}$, we have a collection of
distribution rules denoted by $\mathcal{D}_K$. 
The collection  $\mathcal{D}_K$ 
is the subset of distribution rules
for which $K$ is the value of the secret. The entire set of distribution
rules is denoted by $\mathcal{D} = \cup_K \mathcal{D}_K$.
 We assume without loss of generality that the distribution rules in $\mathcal{D}$ are all distinct. 

When the \emph{dealer} wishes to share a secret  $K \in \mathcal{K}$, they 
first choose 
a distribution rule $d \in \mathcal{D}_K$
and then they use $d$ to distribute shares to the $n$ players. The choice of the secret $K$
and the distribution rule $d$ will be determined by appropriate 
probability distributions. The only property that we will require moving forward is that
every possible distribution rule is used with positive probability (which implies that every
possible secret occurs with positive probability).

\begin{Definition}
\label{weak.rs}
A set of distribution rules $\mathcal{D}$ is a \emph{weak 
$(s,t,n)$ ramp scheme} if the  following two properties are satisfied:

\begin{description}
\item[(1)] Suppose $K,L \in \mathcal{K}$, $|\mathcal{P}_0| \geq t$, $d \in \mathcal{D}_K$, $e \in \mathcal{D}_L$
and $d|_{\mathcal{P}_0} = e|_{\mathcal{P}_0}$. Then $K = L$.
(This property is saying that $t$ or more shares determine a unique secret.)
\item[(2)] Suppose $K \in \mathcal{K}$, $|\mathcal{P}_0| \leq s$ and $d \in \mathcal{D}_K$.
Then, for every $L  \in \mathcal{K}$, there is at least one distribution rule 
$e \in \mathcal{D}_L $ such that $d|_{\mathcal{P}_0} = e|_{\mathcal{P}_0}$.
\end{description}
\end{Definition}

Property {\bf (2)} is saying that $s$ or fewer shares do not rule out any possible value of the secret.
An alternative definition is to require that the probability distribution on the set of 
possible secrets is unchanged even when $s$ shares are known.
Suppose we consider $\mathbf{K}$ to be the random variable defined by the probability distribution 
on the set of secrets $\mathcal{K}$. For any subset of players $\mathcal{P}_0$,  define 
$\mathbf{X}(\mathcal{P}_0)$ to be the random variable determined by the probability distribution
induced on the possible lists (i.e., tuples) of shares given to the players in $\mathcal{P}_0$.

\begin{Definition}
A set of distribution rules $\mathcal{D}$ is a \emph{perfect 
$(s,t,n)$ ramp scheme} if the  following two properties are satisfied:

\begin{description}
\item[(1)] Suppose $K,L \in \mathcal{K}$, $|\mathcal{P}_0| \geq t$, $d \in \mathcal{D}_K$, $e \in \mathcal{D}_L$
and $d|_{\mathcal{P}_0} = e|_{\mathcal{P}_0}$. Then $K = L$.
\item[(2*)] Suppose $|\mathcal{P}_0| \leq s$ and 
let $d$ be any distribution rule.
Then, for every $L \in  \mathcal{K}$, it holds that
\[\mathbf{Pr}[ \mathbf{K} = L  \mid  \mathbf{X}(\mathcal{P}_0) = d|_{\mathcal{P}_0}] 
= \mathbf{Pr}[\mathbf{K} = L].\]
\end{description}
\end{Definition}
Of course, any perfect ramp scheme is also a weak ramp scheme.

\bigskip

The following is a standard method to produce ramp schemes.
We use the presentation of the construction from \cite{OK}
(see also \cite[\S 14.2.1]{Stinson}).

\begin{Example}[Shamir Ramp Scheme]\label{Shamir-ramp}
Let $q \geq n+1$ be a prime power, define $t_0 = t-s$, 
let $\mathcal{K} = (\eff_q)^{t_0}$,
and let $\mathcal{X} = \eff_q$. 
Define $x_1,x_2, \dots ,x_n$  to be $n$ 
distinct non-zero 
elements of $\eff_q$.
The value  $x_i$ is associated with $P_i$, for all $i$, $1 \leq i \leq n$.
There are $q^{t}$ distribution rules in the scheme. For
any $t$-tuple $\mathbf{a} = (a_0, \dots , a_{t-1}) \in (\eff_q)^{t}$
we have an associated distribution rule $d_{\mathbf{a}} \in \mathcal{D}_{(a_0, \dots , a_{t_0-1})}$ defined by
the equation
\[ d_{\mathbf{a}}(P_j) = \sum_{i=0}^{t-1} a_i (x_j)^i, \]
for $1 \leq i \leq n$.
The result is an ideal $(s,t,n)$ ramp scheme with shares from $\eff_q$.
\end{Example}

\begin{Remark}
\label{shamir.rem}
In the construction above, there is a distribution rule corresponding to every possible
polynomial $a(x) \in \eff_q[x]$ of degree at most $t-1$. The shares are evaluations of the polynomial $a(x)$,
analogous to a Reed-Solomon code. The secret consists of the first $t-s$ coefficients of the polynomial $a(x)$.
In the case $t-s = 1$, we have the Shamir threshold scheme \cite{Shamir} and the secret is the
constant term of the the polynomial $a(x)$.  
\end{Remark}

The following result is well-known. We provide the proof, which uses a simple counting argument,  for completeness.

\begin{Lemma}
\label{idealbound} Suppose that $\mathcal{K}$ is the set of possible secrets and $\mathcal{X}$ is the set of
possible shares for any $(s,t,n)$ ramp scheme. Then $|\mathcal{K}| \leq |\mathcal{X}|^{t-s}$.
\end{Lemma} 

\begin{proof}
Let $\mathcal{P}_0 \subseteq \mathcal{P}$ and  $\mathcal{P}_1 \subseteq \mathcal{P}$, 
where $|\mathcal{P}_0| = s$, $|\mathcal{P}_1| = t-s$, and
$\mathcal{P}_0 \cap \mathcal{P}_1 = \emptyset$. Fix any distribution rule $d$.
Property {\bf (2)} states that, for every $L  \in \mathcal{K}$, there exists a distribution rule 
$e_L \in \mathcal{D}_L $ such that $d|_{\mathcal{P}_0} = e_L|_{\mathcal{P}_0}$. Consider the set
\[ \{ e_L|_{\mathcal{P}_1} : L  \in \mathcal{K}\} .\] This set consists of 
$|\mathcal{K}|$ different $(t-s)$-tuples of shares. Therefore, $|\mathcal{K}| \leq |\mathcal{X}|^{t-s}$.
\end{proof}

If equality is achieved in Lemma \ref{idealbound}, then the ramp scheme is termed \emph{ideal}.

\begin{Corollary}
\label{idealTS} Suppose that $\mathcal{K}$ is the set of possible secrets and $\mathcal{X}$ is the set of
possible shares for any $(t,n)$ threshold scheme. Then $|\mathcal{K}| \leq |\mathcal{X}|$.
\end{Corollary} 

\begin{proof}
Set $s = t-1$ in Lemma \ref{idealbound}.
\end{proof}

\section{Combinatorial Equivalences Involving Threshold Schemes}
\label{equiv.sec}

In this section, we review known combinatorial equivalences involving threshold schemes.

\begin{Definition}
An \emph{orthogonal array}, denoted  $\mathrm{OA}(t,k,v)$,
is a $v^t$ by $k$ array $A$ defined on an alphabet $\mathcal{X}$ of cardinality $v$,
such that
any $t$ of the  $k$ columns of $A$
contain all possible $k$-tuples from $\mathcal{X}^t$ exactly once.
\end{Definition}

\begin{Definition}
A \emph{maximum distance separable code}, or \emph{MDS code}, 
of  length $k$ and size $v^t$ over an alphabet $\mathcal{X}$ of size $v$,
is a set of $v^t$ vectors (called \emph{codewords}) in $\mathcal{X}^k$,
having the property that 
the hamming distance between any two distinct codewords 
is at least  $k-t+1$.
\end{Definition}

\begin{Theorem}
\label{equiv1}
The following are equivalent:
\begin{enumerate}
\item an ideal $(t,k-1)$-threshold scheme with shares from an alphabet of size $v$
\item an  $\mathrm{OA}(t,k,v)$
\item an MDS code of length $k$ and size $v^t$ over an alphabet of size $v$.
\end{enumerate}
\end{Theorem}

Given an ideal $(t,k-1)$-threshold scheme, if we write out all the possible 
distribution rules in the form of an array, then this array is the stated orthogonal array.
The rows of the orthogonal array form the codewords in the stated code.

The equivalence of 2.\ and 3.\ is a ``classical'' result that has been known at least since the 
work of Delsarte \cite{Delsarte}. The equivalence
of 1.\ and 2.\ was proven by Keith Martin in 1991 in his PhD thesis  \cite{Mar}, and independently 
by Dawson, Mahmoodian and Rahilly in 1993 \cite{DMR} and by Blakley and Kabatianski
in 1995 \cite{BK}. Much earlier, in 1983, it was stated in 
a paper by Karnin, Greene and Hellman \cite{KGH} that an ideal threshold scheme implies 
the existence of the corresponding MDS code, but their proof is incomplete.

In this paper, we use the term ``linear'' to refer to combinatorial structures that
can be viewed as subspaces of a vector space over a 
finite field\footnote{However, note that the term ``linear ramp scheme''
has a completely different meaning in \cite{JM1}.}. The following theorem 
is the analogue of Theorem \ref{equiv1} restricted to the setting of linear threshold schemes.
First, however, we define some relevant concepts.

\begin{Definition}
Let $q$ be a prime power and let $t\geq 2$. The \emph{desarguesian projective geometry PG$(t-1,q)$}
is based on the vector space $\mathcal{V} = (\eff_q)^t$. The \emph{points} in PG$(t-1,q)$
are the one-dimensional subspaces of $\mathcal{V}$; the \emph{lines} are the  two-dimensional subspaces of 
$\mathcal{V}$, etc. A \emph{hyperplane} in the geometry is a $(t-1)$-dimensional subspace of $\mathcal{V}$.
\end{Definition}

\begin{Definition}
A \emph{$k$-arc} in the projective geometry PG$(t-1,q)$ is a set of $k$ points 
in PG$(t-1,q)$ such that no $t$ of them are on a hyperplane.
\end{Definition}

\begin{Theorem}
\label{equiv2}
The following are equivalent:
\begin{enumerate}
\item an ideal linear $(t,k-1)$-threshold scheme with shares from $\eff_q$
\item  a linear OA$(t,k,q)$ defined over $\eff_q$
\item  a linear MDS code of length $k$ and dimension $t$ over $\eff_q$ 
\item  a $k$-arc in  PG$(t-1,q)$
\item a $t \times k$ matrix $M$ over $\eff_q$ such that any $t$ columns of $M$ are linearly independent.
\end{enumerate}
\end{Theorem}


The equivalence of 1., 2.\ and 3.\ in Theorem \ref{equiv2} is the same as in Theorem \ref{equiv1}, 
since it is clear that the relevant transformations ``preserve'' linearity.

We note that the matrix defined in 5.\ is just a basis for the code or the orthogonal array. 
The points in the $k$-arc are  the one-dimensional subspaces of $(\eff_q)^t$
generated by the columns of this matrix (recall that the ``points'' in a projective geometry of
dimension $d$ are one-dimensional subspaces of a $(d+1)$-dimensional vector space).



The main existence results for these structures come from Reed-Solomon codes (RS-codes),
which are linear MDS codes. The following result, stated in terms of the equivalent 
orthogonal arrays, is well-known
(see, e.g., \cite[Ch.\ 11, \S 5]{MS}).

\begin{Theorem}
\label{RS.thm}
Suppose $q$ is a prime power and $t \geq 2$. Then there is a linear OA$(t,q+1,q)$.
\end{Theorem}

\begin{proof} Let 
$\alpha_1, \dots, \alpha_{q-1}$ be the nonzero elements of $\eff_q$. 
Consider the following $t$ by  $q+1$ matrix $M_0$:

\[ M_0 = 
\left( 
\begin{array}{cccccc}
1 & 1 & 1 & \cdots & 1 & 0\\
0 & \alpha_1 &  \alpha_2 & \cdots & \alpha_{q-1} & 0\\
0 & {\alpha_1}^2 &  {\alpha_2}^2 & \cdots & {\alpha_{q-1}}^2 & 0\\
\vdots &  \vdots & \vdots & \ddots & \vdots & \vdots\\
0 & {\alpha_1}^{t-1} &  {\alpha_2}^{t-1} & \cdots & {\alpha_{q-1}}^{t-1} & 1
\end{array}
\right) .\]
Any $t$ columns of $M_0$ are linearly independent, so 
it generates a linear OA$(t,q+1,q)$
\end{proof}

It has been conjectured (e.g., by Hedayat, Sloane and Stufken, \cite[p.\ 96]{HSS}) that if an 
MDS code of length $k$ and size $q^t$ over $\eff_q$ exists, then 
there is a {\it linear} MDS code with the same parameters.



In the ``general'' case, the first  necessary conditions derive from the 
classical Bush bounds for orthogonal arrays, which were 
proven in 1952 (see \cite{CD}).

\begin{Theorem}[Bush Bound]
\label{bush.thm}
If there is an OA$(t,k,v)$, then
\[ k \leq 
\begin{cases}
v+t-1 & \text{if $t=2$, or if $v$ is even and $3 \leq t \leq v$}\\
v+t-2 & \text{if $v$ is odd and $3 \leq t \leq v$}\\
t+1 & \text{if $t \geq v$.}
\end{cases}
\]
\end{Theorem}

There have been some relatively minor improvements to these general bounds over the years.
On the other hand, the linear case has received considerably more attention
and much more is known in this case.

The following is known as the {\bf Main Conjecture} for linear MDS codes.
It is attributed to Segre (1955).

\begin{Conjecture}[{\bf Main Conjecture}]
Suppose $q$ is a prime power. Let $M(t,q)$ denote the maximum value of $k$ such that there
exists a linear MDS code of length $k$ and dimension $t$ over $\eff_q$.
If $2 \leq t < q$, then
\[ M(t,q) = 
\begin{cases}
q+2 & \text{if $q$ is a power of $2$ and $t \in \{3,q-1\}$}\\
q+1 & \text{otherwise.}
\end{cases}
\]
If $t \geq q$, then $M(t,q) = t+1$.
\end{Conjecture}

The {\bf Main Conjecture} has been shown to be true in many parameter situations, including all  the 
cases where $q$ is prime. This is a famous result of Simeon Ball \cite{Ball} proven in 2012. 

\bigskip

The following
theorem summarizes some of the known results.  These and other
related results are surveyed in \cite{Huntemann}.

\begin{Theorem}
\label{mainconj.thm}
Suppose that $q = p^j$  where $p$ is prime, and suppose $2 \leq t < q$.
Then the {\bf Main Conjecture} is true in the following cases:
\begin{enumerate}
\item $q$ is prime (for all relevant $t$)
\item $q \leq 27$ (for all relevant $t$)
\item $t \leq 5$ or $t \geq q-3$
\item $t \leq p$.
\end{enumerate}
\end{Theorem}

\section{Augmented Orthogonal Arrays}
\label{AOA.sec}

Our objective is to generalize  the results of the previous section
to ideal ramp schemes. It turns out that a certain (apparently new) type of combinatorial array
is required in order to state the resulting equivalences. We define these arrays now, and 
then we prove some basic results about them and give some constructions.

\begin{Definition}
An \emph{augmented orthogonal array}, denoted  $\mathrm{AOA}(s,t,k,v)$,
is a $v^t$ by $k+1$ array $A$ that satisfies the following properties:
\begin{enumerate}
\item the first $k$ columns of $A$ form an orthogonal array $\mathrm{OA}(t,k,v)$ on 
a symbol set $\mathcal{X}$ of size $v$
\item the last column of $A$ contains symbols from a set $\mathcal{Y}$ of size $v^{t-s}$
\item any $s$ of the first $k$ columns of $A$, together with the last column of $A$,
contain all possible $(s+1)$-tuples from $\mathcal{X}^s \times Y$ exactly once.
\end{enumerate}
\end{Definition}

\begin{Remark}
There are generalizations of orthogonal arrays, known as \emph{mixed orthogonal arrays},
that may contain different symbol sets in different columns. See, for example, \cite[Chapter 9]{HSS}.
The AOAs that we have defined above are not mixed orthogonal arrays, however.
\end{Remark}

\begin{Example}
\label{1333.ex}
We give an example of an $\mathrm{AOA}(1,3,3,3)$.
We take $\mathcal{X} = \eff_3$ and $\mathcal{Y} = \eff_3 \times \eff_3$.
The AOA  has the following $27$ rows:
\[ 
\begin{array}{|c|c|c|c|}
\hline
\alpha & \beta & \gamma & (\alpha + \beta, \alpha + \gamma) \\ \hline \end{array}\]
where $\alpha , \beta , \gamma \in \eff_3$.
\end{Example}

The following theorem is immediate from the definitions.

\begin{Theorem}
\label{OA-AOA}
An $\mathrm{AOA}(t-1,t,k,v)$ is equivalent to an $\mathrm{OA}(t,k,v)$.
\end{Theorem}

The next theorem shows an obvious way to construct AOAs from OAs.

\begin{Theorem}
\label{OAtoAOA}
If there exists an $\mathrm{OA}(t,k+t-s,v)$, then there exists an  $\mathrm{AOA}(s,t,k,v)$.
\end{Theorem}

\begin{proof} Merge the last $t-s$ columns of an $\mathrm{OA}(t,k+t-s,v)$ to form
a single column whose entries are $(t-s)$-tuples of symbols.
\end{proof}

We note that the converse of Theorem \ref{OAtoAOA} is not always true.
Example \ref{1333.ex}, which gives a construction of an $\mathrm{AOA}(1,3,3,3)$,
provides an illustration of what can go wrong.

\begin{Example}
The natural way to attempt to construct an $\mathrm{OA}(3,5,3)$ from 
the $\mathrm{AOA}(1,3,3,3)$ presented in Example \ref{1333.ex}
would be to split the last column into two columns
of elements from $\eff_3$. We would get the array
having the following $27$ rows:
\[ 
\begin{array}{|c|c|c|c|c|}
\hline
\alpha & \beta & \gamma & \alpha + \beta & \alpha + \gamma \\ \hline \end{array}\]
where $\alpha , \beta , \gamma \in \eff_3$.
It is easy to see that the fourth column is the sum of the first two columns, 
so these three columns cannot contain all possible $3$-tuples.

In fact, there does not exist any  $\mathrm{OA}(3,5,3)$, because the parameters
violate the Bush bound (Theorem \ref{bush.thm}).
\end{Example}

Later in this section, we will construct some additional examples of $\mathrm{AOA}(s,t,k,v)$
in situations where $\mathrm{OA}(t,k+t-s,v)$ do not exist.

Next, we present an obvious but useful method  to construct
linear AOAs.

\begin{Construction}
\label{linearAOA}
Suppose that $q$ is a prime power. Suppose there is a  $t$ by $k+t-s$ matrix $M$,
having entries from $\eff_q$, which satisfies the following two properties:
\begin{enumerate}
\item any $t$ of the first $k$ columns of $M$ are linearly independent, and 
\item any $s$ of the first $k$ columns of $M$, along with  with the last  
$t-s$ columns of $M$, are linearly independent.
\end{enumerate}
Then there exists a linear AOA$(s,t,k,q)$.
\end{Construction}

The following application of Construction \ref{linearAOA} is a slight generalization of the
Shamir ramp schemes that we already presented in
Example \ref{Shamir-ramp}. 

\begin{Theorem}
\label{Shamir.const}
Suppose $q$ is a prime power and $1 \leq s < t \leq k \leq q$. Then there exists a linear AOA$(s,t,k,q)$.
\end{Theorem}

\begin{proof}
We obtain an AOA$(s,t,k,q)$ using Construction \ref{linearAOA}, 
by defining a suitable $t$ by $k+t-s$ matrix $M$
having entries from $\eff_q$. 

The first $k$ columns of $M$ are obtained by deleting the first column and the last $q-k$ columns from the matrix
$M_0$ defined in the proof of Theorem \ref{RS.thm}. Cal this matrix $M_1$.
The last $t-s$ columns 
of $M$ consist  a $t-s$ by $t-s$  identity matrix and $s$ rows of zeroes:
 \[ M_2 = \left( 
\begin{array}{cccc}
1 & 0 & \cdots & 0\\
0 & 1 & \cdots & 0\\
 \vdots & \vdots & \ddots & \vdots \\
 0 & 0 & \cdots & 1\\
0 & 0 & \cdots & 0\\
 \vdots & \vdots & \ddots & \vdots \\
 0 & 0 & \cdots & 0
\end{array}
\right)  = \left( 
\begin{array}{c}
I_{t-s} \\ \hline \mathbf{0}
\end{array}
\right).\]
Then, construct the matrix $M = \left( \begin{array}{c|c} M_1 & M_2 \end{array} \right)$.
The matrix $M$  satisfies the conditions of Theorem \ref{linearAOA} and 
therefore it yields an AOA having the stated parameters. 
\end{proof}

We are interested in identifying parameters for which there exists an 
AOA$(s,t,k,q)$ but there does not exist an OA$(t,k+t-s,q)$.
Suppose $q$ is odd and $3 \leq t \leq q$. If an OA$(t,k+t-s,q)$ exists,
then Theorem \ref{bush.thm}
asserts that $k+t-s \leq q+t-2$, so $k \leq q+s-2$. Therefore an OA$(t,k+t-s,q)$ does not
exist if $k = q$ and $s=1$. So we have the following.

\begin{Theorem}
\label{nonexist1}
Suppose $q$ is an odd  prime power and $3 \leq t \leq q$. Then there exists an AOA$(1,t,q,q)$
but there does not exist an OA$(t,q+t-1,q)$.
\end{Theorem}

Here is another application of Construction \ref{linearAOA}.

\begin{Theorem}
\label{second.thm}
Suppose there is a linear OA$(t-s,t,q)$. 
Then there exists an AOA$(s,t,t,q)$.
\end{Theorem}

\begin{proof}
Let $N$ be the $t-s$ by $t$ matrix whose rows
form a basis for a linear OA$(t-s,t,q)$. Then the matrix
$M = \left( \begin{array}{c|c} I_t & N^T \end{array} \right)$
satisfies the conditions of Construction \ref{linearAOA} and hence it
yields  an AOA$(s,t,t,q)$.
\end{proof}

We can apply Theorem \ref{second.thm} with $t=q+1$. 
From Theorem \ref{RS.thm}, there is a linear OA$(q+1-s,q+1,q)$ for any $s \leq q-1$ whenever $q$ is a prime power. 
Theorem \ref{second.thm} then 
yields an AOA$(s,q+1,q+1,q)$. However, an  OA$(q+1,2(q+1)-s,q)$ does not exist
because any OA$(q+1,k,q)$ has $k \leq q+2$ by the Bush bound (Theorem \ref{bush.thm})
and $2(q+1)-s \geq q+3$ for $s \leq q-1$.

We have shown the following.

\begin{Theorem}
\label{nonexist2}
Suppose $q$ is a  prime power and $s \leq q-1$. Then there exists an AOA$(s,q+1,q+1,q)$
but there does not exist an OA$(q+1,2(q+1)-s,q)$.
\end{Theorem}

\begin{Example} We can take  $q = 3$, $s = 2$ in Theorem \ref{nonexist2}.
Here we could let 
\[
N = \left(
\begin{array}{cccc}
1 & 1 & 1 & 0 \\
0 & 1 & 2 & 1
\end{array}
\right).
\]
Then we obtain an  AOA$(2,4,4,3)$:
\[
M = \left(
\begin{array}{cccccc}
1 & 0 & 0 & 0 & 1 & 0\\
0 & 1 & 0 & 0 & 1 & 1\\
0 & 0 & 1 & 0 & 1 & 2\\
0 & 0 & 0 & 1 & 0 & 1
\end{array}
\right).
\]
However, there is no OA$(4,6,3)$.
\end{Example}

\section{Equivalence of Ideal Ramp Schemes and AOAs}
\label{ramp-equiv.sec}

We already mentioned that Jackson and Martin \cite{JM1} proved that
strong ideal ramp schemes are equivalent to ideal threshold schemes.
For future reference, we state their result in terms of
orthogonal arrays.

\begin{Theorem}
\cite{JM1} 
\label{equivalence-JM}
There exists a strong ideal $(s,t,n)$ ramp scheme with $v$ possible shares if and only if there exists an
OA$(t,n+t-s,v)$.
\end{Theorem}

Our main goal in this section is to  prove the equivalence of ideal ramp schemes and AOAs. 

\begin{Lemma}
\label{L1}
Suppose $\mathcal{D}$ is the set of distribution rules of an ideal $(s,t,n)$ ramp
scheme having shares from $\mathcal{X}$ and secrets from
$\mathcal{K}$, where $|\mathcal{X}| = v$ and $|\mathcal{K}| = v^{t-s}$. 
Suppose $\mathcal{P}_0 \subseteq \mathcal{P}$ and  $\mathcal{P}_1 \subseteq \mathcal{P}$, 
where $|\mathcal{P}_0| = s$, $|\mathcal{P}_1| = t-s$ and 
$\mathcal{P}_0 \cap \mathcal{P}_1 = \emptyset$. 
Let $d^* \in \mathcal{D}$ and define 
\[ \mathcal{E} = \{ d \in \mathcal{D}: d|_{\mathcal{P}_0} = d^*|_{\mathcal{P}_0} \}.\] 
Then, for every $K \in \mathcal{X}$, 
there is a  distribution rule $d \in \mathcal{D}_K \cap \mathcal{E}$. 
Further, there exists
a bijection $\pi : \mathcal{K} \rightarrow \mathcal{X}^{t-s}$ such that, if
$d \in \mathcal{D}_K \cap \mathcal{E}$, then $d|_{\mathcal{P}_0} = \pi(K)$.
\end{Lemma}

\begin{proof}
Define \[ \mathcal{Z} = 
\{ (d|_{\mathcal{P}_1}, K ) : d \in \mathcal{D}_K \cap \mathcal{E}, K \in \mathcal{K}\}
.\]

\begin{itemize}
\item[(i)]
First, we observe, for every $K \in \mathcal{K}$, that there exists 
$\mathbf{x} \in \mathcal{X}^{t-s}$ such that $(\mathbf{x}, K) \in \mathcal{Z}$
(this follows from property {\bf (2)} of Definition \ref{weak.rs}).

\item[(ii)] 
Next, we note that, if $(\mathbf{x}, K) \in \mathcal{Z}$ and $(\mathbf{x}, L) \in \mathcal{Z}$,
then $K = L$ (this follows from property {\bf (1)} of Definition \ref{weak.rs}).
\end{itemize}

Recall that $|\mathcal{K}| = |\mathcal{X}|^{t-s}$.
Hence  (i) and (ii) imply that, for any $K \in \mathcal{K}$, there is a \emph{unique}
$\mathbf{x} \in \mathcal{X}^{t-s}$ such that $(\mathbf{x}, K) \in \mathcal{Z}$.
This allows us to define 
\[ \pi(K) = \mathbf{x} \Leftrightarrow (\mathbf{x}, K) \in \mathcal{Z},\]
and the function $\pi$ will be a bijection.
%
%
\end{proof}

\begin{Theorem}
\label{main.thm}
If there exists an ideal $(s,t,n)$ ramp scheme defined over a set of $v$ shares, then
there exists an  $\mathrm{AOA}(s,t,n,v)$.
\end{Theorem}

\begin{proof}
Suppose $\mathcal{D}$ is the set of distribution rules of an ideal $(s,t,n)$ ramp
scheme having $v$ possible shares. We now describe how to construct an $(s,t,n,v)$-AOA.
For every $K  \in \mathcal{X}$ and every distribution rule $d \in \mathcal{D}_K$,
construct the $(n+1)$-tuple $r_d = (d_1, \dots , d_n, K)$. We will show that the array 
$\mathcal{A}$ whose rows
are all the $(n+1)$-tuples $r_d$ ($d \in \mathcal{D}$) is the desired augmented orthogonal array.

\begin{enumerate}
\item We first show that  the restriction of $\mathcal{A}$ to any $t$ of the first $n$ columns
of $\mathcal{A}$ consists of
all $v^t$ possible  $t$-tuples, each occurring exactly once.
Without loss of generality, choose the first $t$ columns of $\mathcal{A}$ and let 
$\mathcal{A}'$ denote the
subarray of $\mathcal{A}$ consisting of the first $t$ columns. 

Suppose $r$ is a row of 
$\mathcal{A}'$. This means that there is a distribution rule $d$ such that $d |_{\mathcal{P}_0} = r$.
Suppose that $r'$ is any $t$-tuple that differs from $r$ in a single co-ordinate, say
$r_i \neq r'_i$. Choose any $s$ of the $t-1$ co-ordinates that exclude co-ordinate $i$.
It then follows immediately from Lemma \ref{L1} that there is a 
distribution rule $d'$ such that $d' |_{\mathcal{P}_0} = r'$. Therefore, all $t$-tuples that are
hamming distance $1$ from $r$ are rows in $\mathcal{A}'$. 

This argument can be repeated as
often as desired, to establish that $\mathcal{A}'$ contains all the $v^t$ possible $t$-tuples as rows.
(More formally, we can induct on the hamming distance between $r$ and $r'$.)

Next, we show that no $t$-tuple $r$ occurs more than once as a row of $\mathcal{A}'$.
Suppose that $d |_{\mathcal{P}_1} = d' |_{\mathcal{P}_1} = r$, say, where
$|\mathcal{P}_1| = t$. 
 From property {\bf (1)} of a ramp scheme, we have that $d, d' \in \mathcal{D}_K$, for some $K$.
 The two distribution rules $d$, $d'$ are not identical, so there is a $P_i \not\in \mathcal{P}_1$
 such that $d(P_i) \neq d'(P_i)$. 
Choose any subset $\mathcal{P}_0 \subseteq \mathcal{P}_1$ such that $|\mathcal{P}_0| = s$.
Let $P_j \in \mathcal{P}_1 \setminus \mathcal{P}_0$ and define 
\[ \mathcal{P}_2 = \mathcal{P}_1 \cup \{ P_i \} \setminus \{P_j\}.\]
We have $d |_{\mathcal{P}_0} = d' |_{\mathcal{P}_0}$,
$d, d' \in \mathcal{D}_K$, and $d |_{\mathcal{P}_0 \setminus \mathcal{P}_2} \neq d' |_{\mathcal{P}_0 \setminus \mathcal{P}_2}$,
which contradicts Lemma \ref{L1}.

\item So far, we have shown that the first $n$ columns of $\mathcal{A}$ form a $(t,n,v)$-orthogonal array.
We now have to consider $s$ of the first $n$ columns, together with the last column.
Without loss of generality, take the first $s$ columns of $\mathcal{A}$ and denote them
by $\mathcal{A}''$. 
We already know that there are exactly $v^t$ rows in $\mathcal{A}$.
Given any $s$-tuple $r$, there are $v^{t-s}$ occurrences of $r$ in rows of $\mathcal{A}''$.
By Lemma \ref{L1}, these $v^{t-s}$ rows of $\mathcal{A}''$ 
correspond to all the $v^{t-s}$ different possible values of the secret.
\end{enumerate}

1.\ and 2.\ provide the proof of the desired result.
\end{proof}


The  proof of Theorem \ref{main.thm} applies to both weak and perfect ramp schemes.
We now observe that the AOA yields a
weak or perfect ramp scheme, depending on the probability distributions
that are defined on the set of distribution rules. If we define \emph{any}
probability distribution on the set of distribution rules, we get 
(at least) a weak ramp scheme (provided that every distribution rule is
used with positive probability).
Further, we can ensure that the ramp scheme is perfect by defining 
probability distributions as follows:
\begin{enumerate}
\item Define an arbitrary probability distribution on $\mathcal{K}$
(ensuring that $\mathbf{Pr}[ \mathbf{K} = L ] > 0$ for all $L \in \mathcal{K}$).
\item For any $L \in \mathcal{K}$ and any distribution rule $d \in \mathcal{D}_L$,
define 
$\mathbf{Pr}[ d ] = \mathbf{Pr}[ \mathbf{K} = L ] / v^s$.
\end{enumerate} 
It is straightforward to verify that the resulting ramp scheme is perfect.

Summarizing, we have the following
\begin{Theorem}
\label{equivalence}
\quad 
\begin{enumerate}
\item If there exists a weak ideal $(t,n)$ ramp scheme defined over a set of $v$ shares, then 
there exists an  $\mathrm{AOA}(s,t,n,v)$.
\item If there exists an  $\mathrm{AOA}(s,t,n,v)$, then there exists
a perfect ideal $(t,n)$ ramp scheme defined over a set of $v$ shares. 
\end{enumerate}
\end{Theorem}

We can now identify some parameter situations in which ideal ramp schemes exist, 
but  strong ideal ramp schemes do not exist.

\begin{Theorem}
Suppose $q$ is an odd  prime power and $3 \leq t \leq q$. Then there exists an ideal
$(1,t,q,q)$ ramp scheme but there does not exist 
a strong ideal $(1,t,q,q)$ ramp scheme.
\end{Theorem}

\begin{proof}
This is an immediate consequence of Theorems \ref{nonexist1}, \ref{equivalence-JM} and \ref{equivalence}.
\end{proof}

\begin{Theorem}
Suppose $q$ is a  prime power and $s \leq q-1$. 
Then there exists an ideal
$(s,q+1,q+1,q)$ ramp scheme but there does not exist 
a strong ideal $(s,q+1,q+1,q)$ ramp scheme.
\end{Theorem}

\begin{proof}
This follows from Theorems \ref{nonexist2}, \ref{equivalence-JM} and \ref{equivalence}.
\end{proof}

\section{Summary and Conclusion}
\label{summary.sec}

We showed various parameter situations where 
ideal ramp schemes exist but strong ideal ramp schemes do not exist. 
Our approach was to construct linear AOAs and show that the ``corresponding''
OAs (linear or not) do not exist. It would be easy to find additional
parameter sets for which linear AOAs exist but \emph{linear} OAs do not exist,
by making use of Theorem \ref{mainconj.thm} and other related results in the literature.

All the ramp schemes we constructed in this paper are linear, in the sense that they
are subspaces of vector spaces over a finite field. Constructions of ideal ramp schemes over
alphabets of non-prime power order would also be of interest.

\end{document}